\documentclass[
    twocolumn,
    10pt,
    aps,
    pra,
    superscriptaddress,
    floatfix,
    citeautoscript,
    longbibliography
]{revtex4-2}

\usepackage[T1]{fontenc}
\usepackage[dvipsnames]{xcolor}
\usepackage{amsthm}
\usepackage{thmtools}
\usepackage{mathtools}
\usepackage{braket}
\usepackage{bm} 
\usepackage{bbm} 
\usepackage{dsfont}
\usepackage{amsmath}
\usepackage{amssymb}
\usepackage{enumitem}
\usepackage{float}
\usepackage[colorlinks]{hyperref}
\hypersetup{
    breaklinks  = true,
    colorlinks  = true,
    citecolor   = PineGreen,
    linkcolor   = MidnightBlue,
    urlcolor    = PineGreen
}
\usepackage{inputenc}
\usepackage{tikz}
\usepackage{xurl}

\usepackage{etoolbox}
\apptocmd{\sloppy}{\hbadness 10000\relax}{}{}

\usepackage{silence}
\usepackage{amsfonts,dsfont}
\usepackage{mathrsfs}

\DeclareMathOperator{\opSU}{SU}

\newcommand{\RomN}[1]{\uppercase\expandafter{\romannumeral#1}}

\newcommand{\of}[1]{\left(#1\right)}

\newcommand{\eof}[1]{\left[#1\right]}
\newcommand{\vof}[1]{\left\vert #1 \right\vert}
\newcommand{\nof}[1]{\left\Vert #1 \right\Vert}

\newcommand{\tbE}{\stackrel{!}{=}}

\newcommand{\twoMat}[4]{\begin{pmatrix} #1 & #2 \\ #3 & #4\end{pmatrix}}

\newcommand{\twoArr}[2]{\begin{pmatrix} #1 \\ #2 \end{pmatrix}}

\newcommand{\dd}{\, \text{\upshape d}}

\newcommand{\bR}{\mathbb{R}}

\newcommand{\cT}{\mathscr{T}}

\newcommand{\mIm}{\mathrm{i}}



\usepackage[colorlinks]{hyperref}
\hypersetup{
    colorlinks  = true,
    citecolor   = PineGreen,
    linkcolor   = MidnightBlue,
    urlcolor    = PineGreen
}

\definecolor{fuchsia}{rgb}{0.96, 0.0, 0.63}
\definecolor{dur}{rgb}{0.6, 0, 0.1}

\definecolor{gnuplot_green}{HTML}{009E73}
\definecolor{gnuplot_purple}{HTML}{9400D3}

\declaretheorem[style=plain]{theorem}

\declaretheorem[style=plain]{lemma}
\declaretheorem[style=plain]{remark}

\declaretheorem[style=plain]{definition}
\declaretheorem[style=plain]{example}

\AtBeginDocument{}
\AtBeginDocument{}

\begin{document}

\title{Sub-Riemannian geometry of measurement based quantum computation}

\date{\today}

\author{Lukas \surname{Hantzko}}
\email{lukas.hantzko@itp.uni-hannover.de}
\affiliation{Institut für Theoretische Physik, Leibniz Universit{\"a}t Hannover, Appelstr. 2, 30167 Hannover, Germany}

\author{Arnab \surname{Adhikary}}
\email{arnab.adhikary@itp.uni-hannover.de}
\affiliation{Institut für Theoretische Physik, Leibniz Universit{\"a}t Hannover, Appelstr. 2, 30167 Hannover, Germany}
\affiliation{Department of Physics and Astronomy, University of British Columbia, 
Vancouver, BC V6T 1Z1, Canada}
\affiliation{Stewart Blusson Quantum Matter Institute, University of British Columbia, Vancouver, BC V6T 1Z4, Canada}

\author{Robert \surname{Raussendorf}}
\email{robert.raussendorf@itp.uni-hannover.de}
\affiliation{Institut für Theoretische Physik, Leibniz Universit{\"a}t Hannover, Appelstr. 2, 30167 Hannover, Germany}
\affiliation{Stewart Blusson Quantum Matter Institute, University of British Columbia, Vancouver, BC V6T 1Z4, Canada}

\begin{abstract}

The computational power of quantum phases of matter with symmetry can be accessed through local measurements, but what is the most efficient way of doing so? In this work, we show that minimizing operational resources in measurement-based quantum computation (MBQC) on subsystem symmetric resource states amounts to solving a sub-Riemannian geodesic problem between the identity and the target logical unitary. This reveals a geometric structure underlying MBQC and offers a principled route to optimize quantum processing in computational phases.

\end{abstract}

\maketitle

\section{Introduction}

Quantum computation is an interplay between the continuous and the discrete. For example, the states of any given number of qubits form a continuous set, but their measurement yields discrete outcomes. Universality refers to the continuous groups $\mathrm{SU}(2^n)$, but fault tolerance only provides discrete sets of protected encoded gates \cite{eastinRestrictionsTransversalEncoded2009}. Fortunately, these gates can be composed to circuits dense in $\mathrm{SU}(2^n)$, as demonstrated by the Solovay-Kitaev theorem \cite{bonesteelBraidTopologiesQuantum2005}.

The efficiency of quantum computation is typically based on discrete measures, such as qubit and gate counts, and circuit depth. Yet, a continuous notion also exists \cite{nielsenOptimalControlGeometry2006,nielsenQuantumComputationGeometry2006}, due to Nielsen {\em{et al.}}, and is called geometric complexity. They consider the optimization of execution time for unitary gates, as generated by time-dependent Hamiltonians. Hard-to-implement terms in the Hamiltonian---for example terms that correspond to many-body interactions---are eliminated as generators of gates. The resulting control theory points to sub-Riemannian geometry. In fact, there are numerous examples in optimal quantum control theory where sub-Riemannian geometries are helpful. Most prominent is the implementation of time optimal quantum gates in specific physical scenarios \cite{khanejaTimeOptimalControl2001,khanejaSubRiemannianGeometryTime2002}.

Measurement-based quantum computation (MBQC) \cite{raussendorfOneWayQuantumComputer2001}, in which the process of quantum computation is driven by local measurements, appears to live entirely on the discrete side. For comparison, in the circuit model, continuity arises through unitary evolution generated by suitable Hamiltonians. Quantum gates gradually accumulate over an interaction time  $T$.  The MBQC-simulation of such gates is enacted by projective measurement, which is a discontinuous operation. While the interaction time $T$ of the circuit model has a counterpart in MBQC---the angle that defines the basis of measurement---there is no MBQC-counterpart to the continuous passing of time from 0 to $T$. It therefore seems that the differential-geometric notion of quantum computational efficiency has no bearing on MBQC. 

The opposite is the case, as we explain in this paper. The differential-geometric formalism does apply, and in fact MBQC is a strong case for it: the generators of the executable gates are rigidly specified by symmetry action. There is neither room nor need for a judgement as to which generators of unitaries are easy to implement and which ones are hard. Generators are either permitted or forbidden, and the permitted ones are all equal.

The linking element between differential geometry and MBQC is the phenomenon of computational phases of quantum matter \cite{bartlettQuantumComputationalRenormalization2010,dohertyIdentifyingPhasesQuantum2009,miyakeQuantumComputationEdge2010,elseSymmetryProtectedPhasesMeasurementBased2012,millerResourceQualitySymmetryProtected2015}. It signifies that, in the presence of a suitable symmetry, the computational power of MBQC resource states is uniform across physical phases. Whichever resource state is picked from a given symmetry-protected or symmetry-enriched phase, the set of achievable quantum algorithms is the same. This uniform power may amount to universality or nil, or something in-between, depending on the phase in question.

While computational power is uniform, computational efficiency is not. Moving away from the center of a phase towards its boundary makes MBQC progressively harder. More qubits of the resource state must be spent on every logical operation. The reason for this behaviour is the following: Denote by $\alpha$ the angle between the measurement basis at a given site $i$ and the symmetry-respecting basis (the basis local to $i$ that is element-wise invariant under the symmetry). This measurement implements a logical transformation ${\cal{T}}_i(\alpha)$ which, away from special points in the phase, is not unitary. Rather, it is a noisy channel. The deviation of ${\cal{T}}_i(\alpha)$ from the identity operation is $O(\alpha)$, and the deviation from unitarity is $O(\alpha^2)$ \cite{raussendorfComputationallyUniversalPhase2019}. Therefore, the strategy to achieve a unitary rotation of finite angle $\alpha$ with high accuracy is to split it into $N$ rotations about an angle $\alpha/N$ each. The error resulting from this procedure is $\sim (\alpha/N)^2 N = \alpha^2/N$. It can be made arbitrarily small by increasing $N$, at the expense of consuming more particles of the resource state. 

As the splitting number $N$ is made large, the individual logical channels ${\cal{T}}_i(\alpha/N)$ approach the identity. This invites the differential calculus of \cite{nielsenQuantumComputationGeometry2006}.\smallskip

The technical result that forms the starting point for our analysis, as summarized above, is Corollary~1 of \cite{raussendorfMeasurementbasedQuantumComputation2023} (1D) / Theorem~3 of \cite{herringerMeasurementBasedQuantumComputation2025} (2D). It says that the error $\epsilon$ of an $N$-way split rotation about an angle $\alpha$ is
\begin{equation}\label{scal}
\epsilon \leq \frac{\alpha^2}{N}\frac{1-\sigma^2}{\sigma^2}.
\end{equation}
Therein, $\sigma$, $|\sigma|\leq 1$, is an order parameter that depends on the resource state. 

In addition to its computational significance, the order parameter $\sigma$ has physical meaning. In 1D, the  physical order conducive to MBQC is string order \cite{raussendorfMeasurementbasedQuantumComputation2023}, and correspondingly $\sigma$ is a standard string order parameter. For example, the 1D cluster scenario has a $\mathbb{Z}_2\times \mathbb{Z}_2$ symmetry resulting in two symmetry-protected phases---the non-trivial cluster phase and the trivial phase around the product state $\bigotimes_i |+\rangle_i$ \cite{elseSymmetryProtectedPhasesMeasurementBased2012}. At the cluster point, the string order parameter is  $\sigma=1$. Hence there is no logical error at all. Throughout the cluster phase, $\sigma>0$; hence any error $\epsilon>0$ can be achieved by splitting sufficiently many times $N$. All across the trivial phase the string order parameter vanishes, $\sigma = 0$, and the technique of approximating unitaries by splitting breaks down. The trivial phase has no computational power.

In 2D, the physical orders known to be conducive to MBQC are subsystem symmetry-protected topological and subsystem symmetry-enriched topological physical order; and the corresponding computational order parameters $\sigma$ witness those orders. Computational universality can arise in both scenarios; see \cite{raussendorfComputationallyUniversalPhase2019,devakulUniversalQuantumComputation2018,stephenSubsystemSymmetriesQuantum2019,danielComputationalUniversalitySymmetryprotected2020} and \cite{herringerMeasurementBasedQuantumComputation2025}, respectively.
\smallskip

In the present work, we generalize  Eq.~(\ref{scal}), the fundamental scaling relation for computational phases of quantum matter, from single rotations to entire quantum circuits. Under this generalization, Eq.~(\ref{scal}) remains intact, with the rotation angle $\alpha$ replaced by the Carnot-Caratheodory distance of the unitary evolution simulated; see Theorem~(\ref{thm:implemented_error}) below. In this way, the differential-geometric calculus of quantum computational efficiency \cite{nielsenQuantumComputationGeometry2006} is introduced to measurement-based quantum computation. 

This establishes a rigorous framework for minimizing resource consumption in computational phases of quantum matter. The straightforward approach, namely the gate-wise MBQC-simulation of quantum circuits, is typically unnecessarily costly---for two reasons:
\begin{itemize}
\item[(i)]{The MBQC-native gate sets often have little to do with standard circuit model gate sets, hence simulating the latter with the former causes overhead.}
\item[(ii)]{In computational phases of quantum matter, where individual logical operations need to be close to the identity, it may be more efficient for large gates to `bleed into another' rather than being applied separately in succession.}
\end{itemize}
As an example for (i), the native MBQC gate set for the 2D cluster phase is $\exp(-i\alpha_1/2\, Z_i)$,   $\exp(-i\alpha_2/2\, X_i)$, $\exp(-i\alpha_3/2\, Z_{i-1}X_iZ_{i+1})$, $\exp(-i\alpha_4/2\, X_{i-1}Z_iX_{i+1})$, $\exp(-i\alpha_5/2\, Z_{i-2}X_{i-1}Z_iX_{i+1}Z_{i+2})$, .. \cite{mantriUniversalityQuantumComputation2017} --- not very suggestive a gate set  from the circuit model perspective. Ref \cite{raussendorfSymmetryprotectedTopologicalPhases2017} gives a  first indication that MBQC-native gate sets  can be useful for improving computational efficiency.

An example for (ii) is the $Y$-rotation described in Sec.~\ref{sec:example}; see Fig.~\ref{fig:geodesics} in particular.\smallskip

Our main result, Theorem \ref{thm:implemented_error}, puts the discussion of efficiency in computational phases of matter on a rigorous and versatile basis. It naturally incorporates the MBQC-native gate sets  and continuous computational paths.\medskip

The remainder of this paper is organized as follows. In Sec. \ref{sec:preliminaries}) we review from \cite{raussendorfMeasurementbasedQuantumComputation2023}  the noise model for logical operations in computational phases of quantum matter, and define sub-Riemannian geodesics. With these notions, in Sec. \ref{sec:result} we formulate our main theorem (Theorem \ref{thm:implemented_error}) that extends the error bound for MBQC-implemented gates to arbitrary unitaries. In Sec.~\ref{sec:methods} we prove Theorem~\ref{thm:implemented_error}, and illustrate it through the example of $\mathrm{SU}(2)$. Sec.~\ref{sec:concl} is our Conclusion and Outlook.

\section{Preliminaries}
\label{sec:preliminaries}

In computational phases of quantum matter, non-trivial gates are implemented via measurements in a rotated basis away from the symmetry respecting axis. In case of an ideal resource state---a central state in computational phase considered---each of these non-trivial measurements $M_i$, $i = 1, \dots N$ implements a local target unitary $U_i(\alpha_i) = \mathrm{e}^{\mathrm{i} \alpha_i g_i}$ with $g_i \subset G_i$ the set of locally executable generators and $\alpha_i$ the rotation angle. The set $\mathcal{G} = \bigcup_{i=1}^N G_i$ of executable generators is given by all locally executable generators. The target unitary is implemented as a sequence of rotations $U_i = \mathrm{e}^{\mathrm{i} \alpha_i g_i}$ with angles $\alpha_i$, generators $g_i$ and resulting unitary gate $U = \prod_{i=1}^N U_i$.

Away from the perfect resource states, deeper into the phase, the effect of the local measurements remains a product of quantum channels $C = \prod_{i=1}^N C_i$, under the additional assumption that the symmetry-breaking measurements are far enough apart (longer than the correlation length $\zeta$ of the resource state). The individual logical channels then take the form (see \cite[p. 23]{raussendorfMeasurementbasedQuantumComputation2023})
\begin{equation}\label{chan} 
C_i(\beta_i) = \frac{1 + \sigma}{2} [U_i(\beta_i)] + \frac{1- \sigma}{2} [U_i(-\beta_i)].
\end{equation}
Here $[U]$ denotes the map (super operator) $\rho \mapsto U \rho U^\dagger$. For simplicity, we assume that the computational order parameter $\sigma$ is uniform across all sites. 

The leading effect of non-unit order parameter $\sigma$ is to reduce the rotation angle $\alpha$ of the operation by a factor of $\sigma$, $\alpha_i = \sigma \beta_i$, and we have to adjust the measurement angle $\beta_i$ accordingly. The resulting operation is not exactly unitary, and the resulting error is
$$\epsilon_i = \Vert C_i(\beta_i) - [U_i(\alpha_i)]\Vert_\diamond. $$
The channels $C_i(\beta_i)$ form the sequence of implemented logical operations.
The total error of this implementation is the difference $\mathcal{E} = \prod_{i=1}^N C_i(\beta_i) - \prod_{i=1}^N [U_i(\alpha_i)]$. We measure the errors by their diamond norm and refer to it as the error
$$\epsilon = \Vert \mathcal{E} \Vert_\diamond. $$

In the present setting, we use sub-Riemannian geometry in contrast to Riemannian geometry, because only a small and fixed set of generators is available for the computation. The set of generators is restricted by the computational phase and its symmetries.

If the Lie algebra generated by $\mathcal{G}$ is $\mathfrak{su}(2^n)$, we know that we can execute all unitaries $U\in \mathrm{SU}(2^n)$. 
This is due to Chows Theorem \cite{rovelliChowTheoremStructure2024}, a result from sub-Riemannian geometry applied to the special unitary group $\mathrm{SU}(2^n)$. If $\mathcal{G}$ is not bracket-generating, it generates a Lie subalgebra. The resulting circuits then form some closed Lie subgroup of $\mathrm{SU}(2^n)$.\smallskip

In the following we want to find an error bound with the help of sub-Riemannian geodesics. Those are defined via the following properties:

\begin{definition}[Sub-Riemannian Geodesics, Distance]
\label{definition:SubRiemannianGeodesic}
    Let $\mathcal{G}$ be a set of $n$-qubit Pauli operators. Let $A$ the Lie hull of  the executable generators $\mathcal{G}$ and $G$ the corresponding closed Lie subgroup.
    Then a horizontal curve is a curve $C: [0,1] \to G$ in the unitary group with
    $$ \dot{C}(t) = C(t) c(t), \; c(t) \in \mathrm{span}(\mathrm{i}\mathcal{G}). $$
    Therein, the control $c(t)$ is a linear combination 
    $$ c(t) = \mathrm{i} \sum_{g \in \mathcal{G}} c_g(t) g$$
    of executable generators. The functions $c_g$ are called control functions for the generator $g \in \mathcal{G}$ and form a vector $\vec{c}(t) = (c_g)_{g \in \mathcal{G}}$.
    Given $U \in G$ a target unitary, then a connecting horizontal curve additionally satisfies the boundary conditions
    $$ C(0) = e, \; C(1) = U. $$
    A connecting horizontal curve $C$ is called a geodesic if its length (sub-Riemannian arclength)
    $$ L(C) = \int_0^1 \vert \vec{c}(t) \vert \dd t = \int_0^1 \sqrt{\sum_{g\in\mathcal{G}} c_g^2(t)} \;\dd t $$
    is locally minimal. If it is globally minimal, it is called a minimizing curve. Alternatively one can interpret the length as runtime.
    
    For a minimizing curve $C$, the sub-Riemannian arclength $L(C)$ is the sub-Riemannian (Carnot-Caratheodory) distance $\dd_{\mathrm{CC}}(e,U)$ of $e$ and $U$.
\end{definition}

With these assumptions, minimizing curves exist for all $U$, since horizontal connecting curves exist. In general they are not unique, which can also be seen in the simple example \ref{sec:example}.

\section{Result}
\label{sec:result}

In the following theorem, we generalize the result of \cite[Theorem 1]{adhikaryCounterintuitiveEfficientRegimes2024}. Whereas the existing theorem refers only to unitaries generated by a fixed Pauli operator, the present result refers to all unitaries reachable. 

\begin{theorem}[Implemented Error]
\label{thm:implemented_error}
    Consider a resource state with a translation invariant computational order parameter $\sigma$ and executable generators $\mathcal{G}$.
    Let $A$ be the Lie hull of  the executable generators $\mathcal{G}$ and $G$ be the corresponding closed Lie subgroup.
    Let $N$ be the number of sites with non-trivial measurement bases; i.e., the splitting number. Further assume that the distance between symmetry breaking measurements is much larger than the correlation length, such that their effects are uncorrelated.\\
    Then, the unitary $U \in G$ can be executed on that resource state with the following error bound, 
    \begin{equation}\label{Ebound}
    \epsilon \leq \frac{1}{N} \of{\frac{1}{\sigma^2} - 1} \dd^2_{\mathrm{CC}}(e,U) + \mathcal{O}(N^{-2}).
    \end{equation}
    Therein, $\dd_{\mathrm{CC}}(e,U)$ is the sub-Riemannian distance (Def. \ref{definition:SubRiemannianGeodesic}) of $U$ to the identity $e$.
\end{theorem}

This theorem gives an asymptotic error bound for MBQC on resource states in computational phases of quantum matter. We give an implementation of $U$ on that resource state, taking a sub-Riemannian minimizing curve $C$ and approximating it, using the Lie-Trotter-Suzuki formula. From this perspective, finding the implementation with minimal error is a geometric problem.

The leading  order in the $1/N$ expansion of the error bound $\epsilon$ has a simple structure--it is the product of three factors. The first factor, $1/N$, describes the measurement procedure. Enlarging $N$ can make the error $\epsilon$ arbitrarily small, at the cost of spending computational resources $\sim N$. The second factor describes the resource state used. Its quality only depends on a single quantity, the order parameter $\sigma$. The closer $\sigma$ is to its maximum value of one, the smaller the logical error. The third factor describes the target unitary. It enters through its Carnot-Caratheodory distance from the identity.

The higher order terms $\mathcal{O}(N^{-2})$ also include the error made by approximating the geodesic with finite rotations.\smallskip

{\itshape Outline of the proof.} The critical ingredient in the proof is the connection between error and arc length, provided in Lemma~\ref{lemma:ErrorEnergy}. The proof consists of three steps. First we bound the error of an arbitrary rotation sequence using the triangle inequality (Lem. \ref{lem:error_est_rot_sequence}). 
$$ \epsilon \leq \sum_i \epsilon_i = \of{\frac{1}{\sigma^2} - 1} \sum_i \vof{\alpha_i}^2 + \mathcal{O}(\alpha^3). $$
Second, the error bound is related to the sub-Riemannian distance, as rotation sequences are a special case of horizontal curves (Rem. \ref{remark:CurvesAndError}).
As a third step, we take the Lie-Trotter-Suzuki formula (Def. \ref{definition:FirstOrderLieTrotterSuzuki}) to get a rotation sequence for a minimizing sub-Riemannian geodesic $C$ connecting $e$ and $U$. Finally, we join the three pieces together to get the error bound.

\section{Methods}
\label{sec:methods}

To prove the theorem, we start by calculating the local error. This is the contribution to the error coming from a single measurement. Because the channels are multiplicative, we can use the triangle inequality to get the total error from the local ones.

\begin{remark}[Local Error]
\label{rem:local_error}
    The local error $\epsilon_i$ is minimized by implementing $\beta_i = \frac{\alpha_i}{ \sigma}$ assuming that angles $\alpha_i$ are small. The local error is quadratic in the angle and gets smaller as $\sigma$ approaches $1$
    $$ \epsilon_i = \of{\frac{1}{\sigma^2} - 1} \alpha_i^2 + \mathcal{O}(\alpha_i^3). $$
\end{remark}

\begin{proof}
    We look at the error, that is the difference
    $$ \frac{1 + \sigma}{2} [U_i(\beta_i)] + \frac{1- \sigma}{2} [U_i(-\beta_i)] - [U_i(\alpha_i)]. $$
    Therefore we can expand the super operators $[\mathrm{e}^{\mathrm{i} \alpha g_i}]$ in orders of $\alpha = \alpha_i,\beta_i$, to get
    $$ 1 + \mathrm{i} \alpha [g_i,\cdot] + \frac{1}{2} \alpha^2 [g_i,\cdot]^2 + \mathcal{O}(\alpha^3).$$
    Here, $[g_i,\cdot]$ maps matrices $\rho \mapsto [g_i, \rho]$ to the commutator with $g_i$. Taking the difference, with $\beta_i = \frac{\alpha_i}{ \sigma}$, $\alpha_i^2$ is the leading term, with
    $$\frac{1}{2} \of{\frac{1}{\sigma^2} - 1} \alpha_i^2 \Vert [g_i,\cdot]^2\Vert_\diamond.$$
    The result follows universally for the diamond norm, because $\Vert [g_i,\cdot]^2 \Vert_\diamond = 1$. For other operator norms, like the trace norm, this constant can depend on the number of qubits.
\end{proof}

These local error bounds add up via the triangle inequality to get an approximate error bound.

\begin{lemma}[Error bounds for rotation sequences]
\label{lem:error_est_rot_sequence}
    The total error can be bounded by the sum of the local errors
    $$ \epsilon \leq \sum_i \epsilon_i = \of{\frac{1}{\sigma^2} - 1} \sum_i \alpha_i^2 + \mathcal{O}(\alpha^3). $$.
\end{lemma}

\begin{proof}
    We use induction, triangle inequality and product/composition inequality to get
    $$ \epsilon \leq \sum_{i=1}^N \epsilon_i. $$
    For $N=1$: $\epsilon = \epsilon_1$. From $N$ to $N+1$:
    \begin{equation} \begin{split}
    \epsilon(N+1) & = \nof{ C_{N+1} \circ \prod_{i=1}^{N} C_i - [U_{N+1}] \circ \prod_{i=1}^{N} [U_i] } \\
    & \leq \nof{C_{N+1} \circ \prod_{i=1}^{N} C_i - C_{N+1} \circ \prod_{i=1}^{N} [U_i]} \\
    & + \nof{C_{N+1} \circ \prod_{i=1}^{N} [U_i] - [U_{N+1}] \circ \prod_{i=1}^{N} [U_i]} \\
    & \leq \nof{C_{N+1}} \nof{\prod_{i=1}^{N} C_i - \prod_{i=1}^{N} [U_i]} \\
    & + \nof{C_{N+1} - [U_{N+1}]} \nof{\prod_{i=1}^{N} [U_i]} \\
    & \leq \epsilon(N) + \epsilon_{N+1}
    \end{split}
    \end{equation}
    This estimation is independent of the chosen norm except for the last step, where we impose a normalization $\Vert C_i \Vert = \Vert [U_i] \Vert = 1$. Finally the local error (Rem. \ref{rem:local_error}) can be inserted.
\end{proof}

We also numerically confirmed the error scaling for low-dimensional examples.

\begin{remark}[Numerical Observation]
\label{rem:num_observation}
    Without a lower bound on the implementation error, we cannot conclude that implementation of the geodesic gives the optimal error scaling. In numerical experiments for the diamond norm, we see that this bound is tight (For example Fig. \ref{fig:error_scaling_diamond}). We were however unable to prove that the same expression asymptoticalliy is also a lower bound for the diamond norm of the total error.
\end{remark}

\begin{figure}[hbt!]
    \centering
    \input{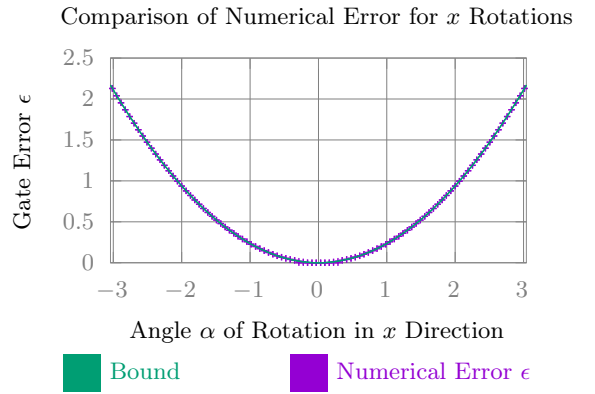}
    \begin{tikzpicture}
        \fill[gnuplot_green] (0,0) rectangle (0.5,0.5) node[below right] {Bound};
        \fill[gnuplot_purple] (3,0) rectangle (3.5,0.5) node[below right] {Numerical Error $\epsilon$};
    \end{tikzpicture}
    \caption{Numerical example for confirmation of the error scaling for a Pauli $X$ rotation. We compute the error numerically (purple) and using our bound (Lemma \ref{lem:error_est_rot_sequence}) (green). As mentioned in the text, the bound is tight for all numerical experiments using also different executable generators. Using the Frobenius or the Trace Norm instead, the bound is not tight for combinations of different generators. Parameters: $\sigma = 0.9$, $N = 500$}
    \label{fig:error_scaling_diamond}
\end{figure}

To find a rotation sequence, that has a low error, we make use of a minimizing horizontal curve in the special unitary group $\opSU\of{2^n}$ which starts at the identity, ends at the desired unitary and satisfies an additional constraint on the instantaneous rotation axis. We first show, that rotation sequences can be interpreted as such curves.

The contribution of the rotation sequence in the total error $\sum_{i=1}^N \alpha_i^2$ can be seen as a discrete version of the integral $\int_0^1 \sum_g c_g(t)^2 \dd t$, which resembles the sub-Riemannian energy (Lemma \ref{lemma:ErrorEnergy}).

\begin{remark}[Curves and Error]
    \label{remark:CurvesAndError}
    The executed gate $U$ from the rotation sequence $\alpha_i,g_i$
    $$U = \prod_{i=1}^N \mathrm{e}^{\mathrm{i} \alpha_i g_i} $$
    is the end-point of a piecewise-one-parameter curve $C: [0,N] \to \mathrm{SU}(2^n)$.
    So for every rotation, we get a piece of the curve $t \in (i-1,i]$:
    $$C(t) = \prod_{1 < j < i} \mathrm{e}^{\mathrm{i} \alpha_j g_j} \mathrm{e}^{\mathrm{i} (t - i + 1)\alpha_i g_i}. $$
    Moreover the piecewise-constant control
    $$ c(t) = \mathrm{i}\, \alpha_i g_i, \quad t \in (i-1,i] $$
    directly give rise to the following relation (substituting $t \to t \cdot N$)
    $$ \epsilon \sim \sum_i \alpha_i^2 = \frac{1}{N} \int_0^1 \vert c(t\cdot N) \vert^2 \dd t + \mathcal{O}(N^{-2}). $$
\end{remark}

Therefore, we conclude, that in order to have a good upper bound for the error, a small constant for the $\frac{1}{N}$ term is necessary. Since $N$ will be sufficiently large, only the first order coefficient is influenced by the choice of rotation sequence.
$$ \lim_{N\to \infty} N \cdot \epsilon(N) = \of{\frac{1}{\sigma^2} -1 } \int_0^1 \vert c(t) \vert^2 \mathrm{d} t. $$
By the following lemma we know that minimizing this scaling factor is equivalent to minimizing the sub-Riemannian arclength: 

\begin{lemma}[Error and Energy]
    \label{lemma:ErrorEnergy}
	Sub-Riemannian geodesics minimize the energy functional \cite[262]{agrachev2013control}.
    \begin{equation}
        E(C) = \int_0^1 \vert c(t) \vert^2 \dd t.
    \end{equation}
    The minimum is given by $E = T^2$ for $T = L(C)$ the length of $C$.
\end{lemma}

This step directly relates the sub-Riemannian problem of minimizing some functional to the physical problem of minimizing the error of a MBQC scheme. If the rotation angles could be chosen continuously, a curve minimizing the arclength would be the desired rotation sequence. If we were further able to prove a lower bound for the error, this would imply that the minimization problems are equivalent.

\begin{proof}
    Reparameterize $\vert c(t) \vert = 1$ by arclength. Since $E$ is finite, $\vert c \vert$ is square integrable.
    $$ T^2 = \left(\int_0^T \vert c(t) \vert \dd t\right)^2 = T \int_0^T \vert c(t) \vert^2 \dd t. $$
    Here, we applied the Cauchy-Schwarz inequality with $g(t) = 1$. Moreover, since $\vert c(t) \vert$ is constant, equality holds. Therefore minimizing the arclength (the square $T^2$) is equivalent to minimizing the leading $\frac{1}{N}$ coefficient of the error bound. 
\end{proof}

Since the control cannot be applied continuosly, we approximate the geodesics for a good error scaling using a Trotter-like product formula for the pieces of the geodesic. These need to vanish faster in terms of $N$, so they can be neglected compared to the implementation error.

There are some choices for higher-order Trotter like product formulae, which also include the commutator and higher-order error \cite[5]{Chen_2022}. One can also use \cite{strichartz1987campbell} for the time-ordered exponential, or other splitting techniques \cite{casasApproximatingExponentialsCommutators2025}.

We adapt a scheme from \cite{wiebe2010higher} to get a local approximation of third order. This results in a global approximation error of second order, which is sufficient, because the error in the implementation is of first order, so it is dominant.
We recall from \cite{wiebe2010higher}, Page 5 the following definition.

\begin{definition}[First order Lie-Trotter-Suzuki product formula] 
\label{definition:FirstOrderLieTrotterSuzuki}

Let $c$ be a time-dependent control
$$ c\of{t} = \mathrm{i} \sum_{g \in \mathcal{G}} c_g\of{t} g. $$
We label the generators $g_i \in \mathcal{G}$ with $i = 1, \dots, m$. Then its first order Lie-Trotter-Suzuki formula is
\begin{equation}
\begin{split}
    \label{eq:firstorderLieTrotterSuzuki}
    & U_1\of{t + \Delta,t} \equiv \\
    & \prod_{i=1}^m \exp \left[\mathrm{i } c_i\of{t + \frac{\Delta}{2}}\frac{\Delta}{2}\right] \prod_{i=m}^1 \exp \left[\mathrm{i } c_i\of{t + \frac{\Delta}{2}}\frac{\Delta}{2}\right.
\end{split}
\end{equation}
\end{definition}

In figure \ref{fig:geodesics} we compare the resulting curve for the rotation sequence generated by the Lie-Trotter-Suzuki formula to the original geodesic. The local approximation order (Lemma \ref{lemma:localhigherordertrotter}) is also given in \cite{wiebe2010higher}, Page 9.

\begin{figure*}[hbt]
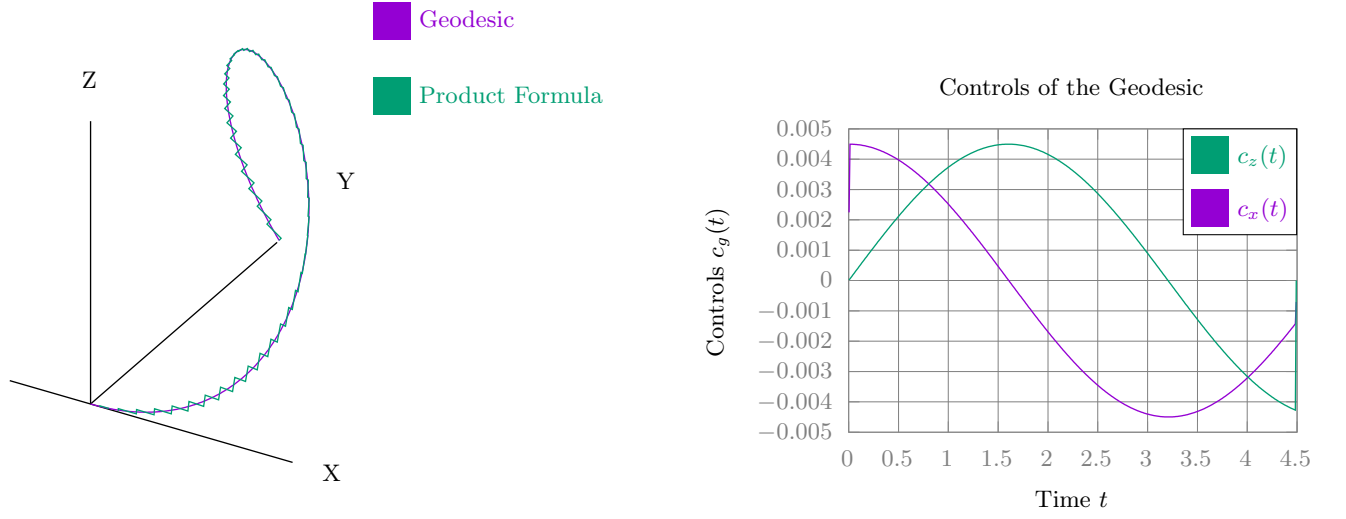

    \begin{minipage}[t]{0.45\linewidth}
    \hspace{-1.5cm}
    \input{Figures/gnuplot_trotter_y.tex}
    \begin{tikzpicture}[overlay]
    \begin{scope}[xshift=-1.76cm,yshift=5.36cm]
        \fill[gnuplot_green] (0,0) rectangle (0.5,0.5) node[below right] {Product Formula};
        \fill[gnuplot_purple] (0,1) rectangle (0.5,1.5) node[below right] {Geodesic};
    \end{scope}
    \end{tikzpicture}
    \end{minipage}
    \hfill
    \begin{minipage}[t]{0.45\linewidth}
    \input{Figures/gnuplot_controls.tex}
    \begin{tikzpicture}[overlay]
    \begin{scope}[xshift=2.55cm,yshift=4.95cm]
        \draw[fill=white] (-0.1,0.6) rectangle (1.4,-0.8);
        \fill[gnuplot_green] (0,0) rectangle (0.5,0.5) node[below right] {$c_z(t)$};
        \fill[gnuplot_purple] (0,-0.7) rectangle (0.5,-0.2) node[below right] {$c_x(t)$};
    \end{scope}
    \end{tikzpicture}
    \end{minipage}
	\caption{The sub-Riemannian geodesic (Example \ref{example:SubRiemannianGeodesics}) for a rotation in the $y$ direction (Purple). For that curve we compare the Lie-Trotter-Suzuki product formula (Green), which is a curve with piecewise constant control. To the right, the control functions $c_z(t)$ (purple) and $c_x(t)$ (green) are depicted.}
    \label{fig:geodesics}
\end{figure*}

\begin{lemma}[Local approximation order]
    \label{lemma:localhigherordertrotter}
    Let $U = \cT\exp\eof{\int\dd t\, c\of{t}}$ be the unitary generated by the continuous time-dependent control $c\of{t}$.
    The first-order Lie-Trotter-Suzuki product formula \eqref{eq:firstorderLieTrotterSuzuki} has a local error of third order. In other words $\nof{U - U_1} \in \mathcal{O}\of{\Delta^3}$.
\end{lemma}

This lemma is important to argue that going to the domain of continuous curves and approximating them does not drastically change the metric of interest, which is of first order in $N^{-1}$.

\begin{proof}
    We start Taylor expanding both parts, the exact rotation using \cite[Lemma 1, Lemma 7]{wiebe2010higher}
    \begin{equation}
        U\of{t + \Delta,t} = 1 + c\of{t} \Delta + \eof{c^2\of{t} + c'\of{t}} \Delta^2/2 + \mathcal{O}\of{\Delta^3},
    \end{equation}
    and the approximation by $c\of{t} \equiv c\of{t + \Delta/2}\Delta$. There, using the expansion of $c\of{t + \Delta/2} \approx c\of{t} + c'\of{t}\Delta/2$
    \begin{equation}
    \begin{split}
        & \exp\of{c\of{t+\Delta/2}\Delta} \\
        & = 1 + c\of{t + \Delta/2} \Delta + c^2\of{t + \Delta/2} \Delta^2/2 + \mathcal{O}\of{\Delta^3} \\
        & = 1 + c\of{t} \Delta + \eof{c^2\of{t} + c' \of{t}} \Delta^2/2 + \mathcal{O}\of{\Delta^3}
    \end{split}
    \end{equation}
    The Lie-Trotter-Suzuki product formula gives rise to the same in second order
    \begin{equation}
    \begin{split}
        & U_1\of{t + \Delta,t} \\
        & = 1 + c\of{t + \Delta/2}\Delta + c^2\of{t + \Delta/2} \Delta^2/2 + \mathcal{O}\of{\Delta^3}
    \end{split}
    \end{equation}
    So for any operator norm, due to homogeneity and triangle inequality: $\nof{U_1 - U} \in \mathcal{O}\of{\Delta^3}$.
\end{proof}

We summarize this, stating how a unitary $U$ is implemented on a resource state with $N\cdot \Delta$ qubits.

\begin{definition}[Curve Implementation]
\label{def:CurveImplementation}
    Let $U\in \mathrm{SU}(2^n)$ the target unitary and let $C(t)$ be a minimizing geodesic with control $c(t)$ such that $C(T) = U$. Split the geodesic into $M$ pieces with $M = \lfloor\frac{N}{(2m-1)}\rfloor$. Then we get a rotation sequence containing less than $N$ individual rotations with angles of order $\mathcal{O}(N^{-1})$, according to (Eq. \ref{eq:firstorderLieTrotterSuzuki}).
\end{definition}

Using this definition, we can prove our main theorem (Thm. \ref{thm:implemented_error}).

\begin{proof}[Proof of Theorem \ref{thm:implemented_error}]
    All the generators in the rotation sequence are executable $c_g(t_i) g$, because we use a horizontal curve $c(t) \in \mathrm{i}\;\mathrm{span} \,\mathcal{G}$. \\
    The total error is bounded by the implementation error and the approximation error
    $$ \nof{\prod_{i=1}^N C_i - [U]} \leq \nof{\prod_{i=1}^N C_i - \prod_{i=1}^N [U_i]} + \nof{\prod_{i=1}^N [U_i] - U}$$
    The approximation error is of order $N^{-2}$ (using triangle inequality):
    $$ \nof{ \prod_i \mathrm{e}^{\mathrm{i} \alpha_i g_i} - U} \leq \sum_i \mathcal{O} (N^{-3}) \in \mathcal{O}(N^{-2}). $$
    For the implementation error we have similar to numerical integration up to second order $\mathcal{O}(N^{-2})$,
    $$\sum_i \vert\alpha_i\vert^2 \leq \frac{1}{N}\int_0^1 \vert c(t) \vert^2 \mathrm{d} t = \frac{1}{N} \dd^2_{\mathrm{CC}}(e,U). $$
    With the last equality holding for a minimizing curve.
\end{proof}

\begin{figure}[hbt!]
    \centering
    \input{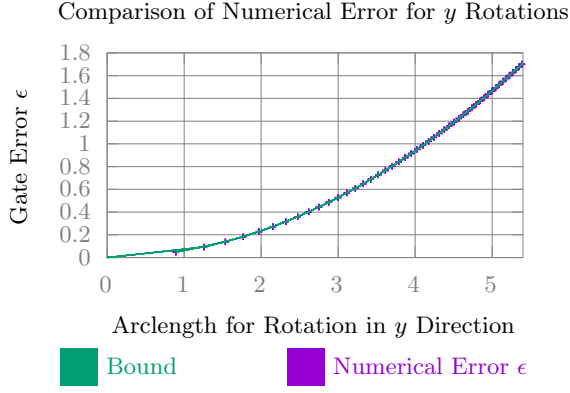}
    \begin{tikzpicture}
        \fill[gnuplot_green] (0,0) rectangle (0.5,0.5) node[below right] {Bound};
        \fill[gnuplot_purple] (3,0) rectangle (3.5,0.5) node[below right] {Numerical Error $\epsilon$};
    \end{tikzpicture}
    \caption{We compare the error norm $\epsilon$ for our implementation of a $y$ rotation (purple) to the theoretical upper limit (green) given by theorem \ref{thm:implemented_error}. We see that the error satisfies the bound. Parameters: $\sigma = 0.9$, $N = 500$}
\end{figure}

\subsection{Scaling}
\label{sec:scaling}

Since computing the sub-Riemannian distance relies on finding shortest geodesics, the computation of the distance and some optimal control is hard.
The advantage of the following $\infty$-distance in comparison to the Carnot-Carathéodory (CC) distance is, that it can be calculated more easily for arbitrary unitaries. Basically only the degree of the used generators is of importance.
Let $X_i$ an orthonormal basis of the Lie algebra. The degree $\deg X_i$ of $X_i$, i.e. the number of commutators needed to generate $X_i$ from $\mathcal{G}$.
$$ y = \exp\left(\sum_i y_i X_i\right)$$
then $\infty$-distance to the identity is
$$d_\infty(e,y) = \max_{i=1,\dots,N} \left\{\vert y_i\vert^{\frac{1}{\deg X_i}}\right\}.$$

The box-ball theorem says that sub-Riemannian metrics are bi-Lipschitz equivalent \cite[126]{nagelBallsMetricsDefined1985}, in the following way:
There are constants $0 < C_1 \leq C_2 < \infty$ such that for all $U \in G$,
$$ C_1 \dd_\infty(e,U) \leq \dd_{\mathrm{CC}}(e,U) \leq C_2 \dd_\infty(e,U). $$
So for the scaling of the sub-Riemannian distance we only need to know the scaling of the $\infty$ metric. The inequality for the scaling $\dd_{\mathrm{CC}} \leq C \vert y_i \vert^\frac{1}{N}$ can be shown using the BCH formula and is sometimes stated as a consequence or a by-product of Chows theorem for reachability \cite[p. 28]{rovelliChowTheoremStructure2024}.

The downside of the result is that the constants $C_1$ and $C_2$ need not be easy to determine or even small. Moreover, they will also depend on the dimension of the Lie group and the executable generators. However, for small angles the bound gets more and more tight. For prominent sub-Riemannian geometries such as the Heisenberg group, explicit constants can be given \cite{bellaicheTangentSpaceSubRiemannian1997}.

Another way to think about upper bounds is by constructing other curves that implement unitaries with the same restriction. One could implement arbitrary unitaries using a generalization of Euler angles. In \cite{smithOptimallyGeneratingSu2N2024}, they find rotation sequences for exponentials of single Paulis with length $t+(\mathrm{deg} X_i -1) \pi$. They also use that in MBQC $\frac{\pi}{2}$ rotations by Pauli operators can be done by classical post-processing.
$$ \mathrm{e}^{\mathrm{i}\frac{\pi}{2} P_1} \dots \mathrm{e}^{\mathrm{i}\frac{\pi}{2} P_{L-1}} \mathrm{e}^{\mathrm{i} t P_L} \mathrm{e}^{-\mathrm{i}\frac{\pi}{2} P_{L-1}} \dots \mathrm{e}^{-\mathrm{i}\frac{\pi}{2} P_1}. $$
Implementing that sequence gives long sequences of a single type of rotation and introduces an additive overhead to implement that sequence for every single Pauli generator appearing in the unitary.
That means, that especially for unitaries close to the identity, the scaling argues in favor of the sub-Riemannian minimizing curves.
In figure \ref{fig:box-ball-su2} we see as an example that even for a single generator the Euler angles are inefficient compared to using geodesic.
With regards to the executable generators we come to the same conclusion as \cite{smithOptimallyGeneratingSu2N2024}:
Optimal generation of the Lie algebra using commutators is beneficial for optimal generation of unitaries using only executable generators. This is due to the scaling of the length of both measures with the degree of the generators.

\subsection{Simplest Example}
\label{sec:example}

In general, the problem of finding sub-Riemannian geodesic is not easy to solve.
For the case of $\opSU\of{2}$, there is an example problem that is given in \cite[263]{agrachev2013control} with the following curves:

\begin{example}[Sub-Riemannian geodesics]\label{example:SubRiemannianGeodesics}
    Compare \cite[263-265]{agrachev2013control}.
    Consider $V$ as the span of $\mIm/2 \sigma_x, \mIm/2 \sigma_z$, the control parametrized as $c\of{t} = \mIm/2\of{c_x\of{t} \sigma_x + c_z\of{t} \sigma_z}$.
    The sub-Riemannian geodesics have time-dependent coefficients $c\of{t}$ of the form
    \begin{equation}
    \begin{split}
        \label{eq:rotAxis}
    	& c:\eof{0,T} \to \bR^2; \\
        & \vec{c}\of{t} \stackrel{\cdot}{=} \twoMat{\cos\phi_0}{\sin\phi_0}{-\sin\phi_0}{\cos\phi_0}\twoArr{\cos \beta t}{\sin \beta t} = \twoArr{\cos \beta t + \phi_0}{\sin \beta t + \phi_0}.
    \end{split}
    \end{equation}
    with $\phi_0,\beta$ and $T$ constants to be determined by setting equal the curve at time $T$ to the target unitary $U$
    \begin{equation}
    \begin{split}
    	& U \tbE C\of{T} \\
        & = \exp\eof{\frac{\mIm}{2} T \of{\cos\phi_0 \sigma_x + \sin\phi_0 \sigma_z + \beta \sigma_y}} \exp\eof{-\frac{\mIm}{2} T \beta \sigma_y}.
    \end{split}
    \end{equation}
\end{example}

Given the rotation to implement, the parameters of the geodesic can be found. For most cases, this involves solving a nonlinear equation for $\beta$. In (Fig \ref{fig:geodesics}) there are examples for a geodesic with a unitary $\mathrm{e}^{\mIm\alpha/2 \sigma_y}$.

Finally, we give examples of the bounds on the sub-Riemannian distance on $\opSU(2)$ for the comparison of the different implementations (Fig. \ref{fig:box-ball-su2}). We see, that the decomposition into Euler angles gives a larger arclength than the sub-Riemannian geodesics. Moreover, we see that the $\sqrt{\alpha}$ scaling from the ball-box theorem gives good bounds.

\begin{figure}[hbt!]
    \centering
    \input{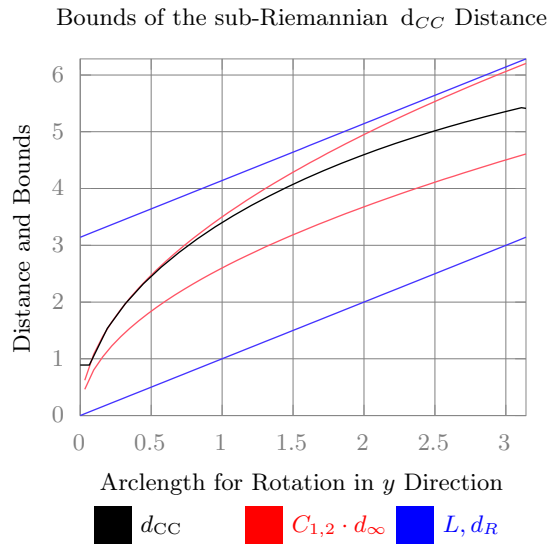}
    \begin{tikzpicture}
        \fill[black] (0,0) rectangle (0.5,0.5) node[below right] {$d_{\mathrm{CC}}$};
        \fill[red] (2,0) rectangle (2.5,0.5) node[below right] {$C_{1,2}\cdot d_{\infty}$};
        \fill[blue] (4,0) rectangle (4.5,0.5) node[below right] {$L,d_{R}$};
    \end{tikzpicture}
    \caption{We compare the different Bounds for a $Y$-rotation in the $\mathrm{SU}(2)$ with $X$ and $Z$ as horizontal directions. Riemannian distance or arclength and scaling for the Euler angles (Blue). The Sub-Riemannian distance that is important to our error bound (Black). Bounds with $d_\infty$ given by the box-ball theorem (Red).}
    \label{fig:box-ball-su2}
\end{figure}

\section{Conclusion and Outlook}
\label{sec:concl}

This work establishes that sub-Riemannian minimizing curves provide a principled and effective method for realizing quantum algorithms with rigorously bounded logical error in measurement-based quantum computation, specifically in the regime of computational phases of quantum matter. By relating computational efficiency directly to the geometry of the executable generators, we uncover a profound structural connection between symmetry-constrained quantum operations and differential geometry. The resulting error scaling outperforms traditional strategies such as Euler angle decompositions.

 The presented methodology applies naturally to more complex Lie groups and richer MBQC architectures. In particular, it applies to multi-generator gate sets governed by subsystem symmetries in 2D and higher-dimensional systems. This contains the 2600 computational phases identified in \cite{herringerClassificationMeasurementbasedQuantum2023}. These examples, with many more to 
 be enumerated and analysed, and perhaps even fully classified, form a rich and varied phenomenology to which a future geometric theory of computational efficiency in quantum phases of matter can refer.

Numerical methods play a complementary role in exploring this broader landscape. For initial-value problems, geodesics can be obtained by solving Hamilton’s equations. For boundary-value formulations, one may use Monte Carlo sampling \cite{dasilvaMonteCarloApproach2025} or the Krotov method \cite{goerzKrotovPythonImplementation2019}, though both are computationally demanding. A more scalable alternative is the fast marching algorithm \cite{sanguinettiSubRiemannianFastMarching2015,craneSurveyAlgorithmsGeodesic2020}, which offers efficient approximations to minimal curves and could serve as a practical tool for large-scale optimization in MBQC contexts.

Several questions remain open. First, can an analogous error minimization framework be developed for the correlated regime of MBQC ~\cite{adhikaryCounterintuitiveEfficientRegimes2024,Adhikary_2021}, where operations are not independent and inter-site correlations play a critical role?--- Ref. ~\cite{raussendorfMeasurementbasedQuantumComputation2023} provided the first scaling relation of the type of Eq.~(\ref{Ebound}), but under the limitations of applicability (a) to only the uncorrelated regime, and (b) to only single generators. Ref.~\cite{adhikaryCounterintuitiveEfficientRegimes2024} lifts restriction (a), and the present work lifts restriction (b). It is desirable to have a scaling result that lifts both restrictions simultaneously.

Second, in all scenarios we analyzed numerically, in the large-$N$ limit, our error bound Eq.~(\ref{Ebound}) not only bounds but equals the actual error. This prompts the question: Can a lower bound on $\epsilon$ be established that matches the upper bound at leading order $1/N$?
This would establish the asymptotic optimality of the geodesic construction. Proving such a bound would formally equate the sub-Riemannian minimization problem with the most efficient possible protocol under physical constraints.

The broadest and arguably richest question in the present context is after the geometric theory of computational efficiency of MBQC. It would answer questions like: Which resource states lead to the most efficient MBQC-implementations of quantum algorithms? For which quantum algorithms and subroutines do the MBQC-native gate sets  provide an advantage?---Ref.~\cite{stephenUniversalMeasurementBasedQuantum2024} has hinted that significant gains can be encountered. The present work provides a foundation for an efficiency analysis of MBQC, for the continuous case of computational phases of quantum matter.

\begin{acknowledgments}
    LH acknowledges helpful discussions with Tom Tischel, Paul Herringer, Wolfram Bauer, Hendrik P. Nautrup and Vir B. Bulchandani.
    This work was supported by the Alexander von Humboldt Stiftung.
\end{acknowledgments}

\section*{Data Availability Statement}
The depicted data as well as the source code are publicly available on GitHub: \cite{hantzkoSimulationGeodesics2025}.

\bibliographystyle{apsrev4-2}
\bibliography{literature.bib}

\onecolumngrid

\end{document}